\RequirePackage{fix-cm}
\documentclass[smallextended]{svjour3}       % onecolumn (second format)
\smartqed  % flush right qed marks, e.g. at end of proof
\usepackage{graphicx}
\usepackage[usenames,dvipsnames]{xcolor}
\usepackage[normalem]{ulem}
\begin{document}

\title{The decoupling problem of the Proca equation; and treatment of Dirac, Maxwell and Proca fields on the resulting pp-wave spacetimes%\thanks{Grants or other notes
%about the article that should go on the front page should be
%placed here. General acknowledgments should be placed at the end of the article.}
}
%\subtitle{Do you have a subtitle?\\ If so, write it here}

\titlerunning{Dirac Maxwell and Proca fields on pp-wave spacetimes}        % if too long for running head

\author{Koray D\"{u}zta\c{s}       \and \.{I}brahim Semiz
        %Second Author %etc.
}

%\authorrunning{Short form of author list} % if too long for running head

\institute{K. D\"{u}zta\c{s}  \and \.{I}brahim Semiz  \at
            Bo\u{g}azi\c{c}i University, Department of Physics \\ Bebek 34342, \.Istanbul, Turkey \\
              \email{duztasko@hotmail.com} \email{semizibr@boun.edu.tr}           %  \\
%             \emph{Present address:} of F. Author  %  if needed
           %\and
          % S. Author \at
             % second address
}

\date{Received: date / Accepted: date}
% The correct dates will be entered by the editor

\maketitle

\begin{abstract}
In this work we take a formal approach to the problem of decoupling Proca equations in curved space-times. We use Newman-Penrose (NP) two-spinor formalism to represent the Proca vector by one complex and two real scalars. We show that a decoupled second order differential equation  for one of the real scalars can be derived if and only if the background space-time admits a covariantly constant null vector. Thus, the background space-time must be a pp-wave vacuum. We evaluate the separability of Proca, Maxwell and Dirac equations on the resulting pp-wave background.
% \PACS{PACS code1 \and PACS code2 \and more}
% \subclass{MSC code1 \and MSC code2 \and more}
\end{abstract}

\section{Introduction}
Wave equations representing massive and massless fields in curved space-times have been extensively studied since 1970's. The solutions describing the propagation of fields are required in various problems including --but not limited to-- black hole stability, Hawking radiation and scattering and absorption of test fields. In an analytic approach  it is essential to decouple individual degrees of freedom and achieve separation of variables, so that the system can be reduced to a set of decoupled ordinary differential equations. In that context, massive scalar field equation in Kerr background was first separated by Matzner~\cite{matzner}. Unruh decoupled and separated massless Dirac equation in Kerr background~\cite{unruhsup}. Teukolsky decoupled massless field equations for scalar, neutrino, electromagnetic, and gravitational fields and derived a separable master equation parametrized to represent each case~\cite{teuk1}. Chandrasekhar decoupled and separated massive Dirac equation in Kerr background~\cite{chandradirac}. Page decoupled and separated Dirac equation in Kerr-Newman background which represents  a charged, rotating black hole~\cite{page1}. The wave equation of a massive real scalar field was also separated in Kerr-Newman background~\cite{rowan}. The wave equation representing Rarita-Schwinger field (spin 3/2) was decoupled and separated by G\"uven in Kerr background~\cite{guven1}, later this was generalised to Kerr-Newman space-time~\cite{guven2}. Separation of variables for the equation of a complex massive scalar field in Kerr-Newman background was also studied~\cite{couch}. This was generalised to the case of dyonic black holes which can have magnetic charge, by Semiz~\cite{semizscalar}. Massive complex Dirac equation in dyonic Kerr-Newman background was also decoupled and separated by Semiz~\cite{semizdirac}.

It turns out that the equations describing all massive and massless fields can be decoupled and separated in Kerr background except massive spin 1 or Proca field.  The longitudinal degree of freedom acquired by spin 1 field when the mass term is introduced, prevents the full decoupling of different components of the field in black hole space-times. For that reason authors have recently employed numerical techniques to solve separable but coupled wave equations for massive vector fields in Schwarzschild background~\cite{proca1,proca2}, where one degree of freedom can be decoupled only.

In this paper we take a formal approach to study the problem of decoupling the Proca equation. Using Newman Penrose (NP) two spinor formalism~\cite{newpen}, we derive the conditions that should be satisfied by space-times, to enable the decoupling of Proca equation; and outline the treatment of Dirac, Maxwell and Proca fields in this context. 

\subsection{Vectors and spinors}
Let $(o, \iota)$ be a spin-basis for a two dimensional vector space $S$ over complex numbers $\mathcal{C}$, endowed with a symplectic structure $\epsilon_{AB}=-\epsilon_{BA}$. The condition that $(o,\iota)$ is a spin basis for $S$ gives
\begin{equation}
\begin{array}{l}
\epsilon_{AB} o^A o^B=\epsilon_{AB}\iota^A \iota^B=0 \\
\epsilon_{AB}o^A \iota^B=1
\end{array}
\end{equation}
Every spin basis induces a tetrad of null vectors.
\begin{equation}
l^a=o^A \bar{o}^{A'} \quad n^a=\iota^A \bar{\iota}^{A'} \quad m^a=o^A \bar{\iota}^{A'} \quad \bar{m}^a=\iota^A \bar{o}^{A'} \label{nptetrad}
\end{equation}
$l$ and $n$ are real while $m$ and $\bar{m}$ are complex (conjugates). The NP null tetrad satisfies orthogonality relations
\begin{eqnarray}
& &l_an^a=n_al^a=-m_a \bar{m}^a=-\bar{m}_a m^a=1 \nonumber \\
& &l_am^a=l_a \bar{m}^a=n_am^a=n_a\bar{m}^a=0
\label{ortog}
\end{eqnarray}

A hermitian spinor $\Theta$ is defined to be one for which $\bar{\Theta}=\Theta$. For this to make sense $\Theta$ must have as many primed as unprimed indices and relative positions of the unprimed indices must be the same as the relative positions of the primed ones. Hermitian spinors can be identified with the tangent space (at a point) of a 4 dimensional manifold $M$. The simplest example is $\Theta^{AA'}$. There exist scalars $\xi$, $\eta$, $\zeta$ such that
\begin{equation}
\Theta^{AA'}=\xi o^A \bar{o}^{A'}+\eta i^A \bar{i}^{A'}+\bar{\zeta} o^A \bar{i}^{A'}+\zeta i^A \bar{o}^{A'}
\label{spinvec}
\end{equation}
As a consequence of the requirement that  $\Theta^{AA'}$ is hermitian, $\xi$ and $\eta$  real. Thus, the set of hermitian spinors forms a real vector space of dimension 4. In spinor analysis we  identify this with $T_p (M)$, the tangent space at a point in a 4-manifold $M$. Similarly the set of hermitian spinors  $\Theta_{AA'}$ forms a real vector space dual to the one above, which will form the cotangent space $T_p^* (M)$. This identification, --which is denoted by relabelling $AA'\mapsto a$-- enables the expression of tensors in the language of spinors (see e.g. \cite{agr,penrosebook}). 

\subsection{The spinor covariant derivative}
The spinor covariant derivative is defined axiomatically as a map $\nabla_a=\nabla_{AA^{\prime}}: \theta_{...} \mapsto \theta_{...;A A^{\prime}}$, where $\theta$ is any spinor field.  In the NP formalism, derivatives are projected onto the null tetrad, hence conventional symbols are defined for those projections,
\begin{equation}
D=l^a\nabla_a,\quad \Delta=n^a\nabla_a,\quad \delta=m^a\nabla_a,\quad \bar{\delta}=\bar{m}^a\nabla_a
\end{equation}
and $\nabla_a$ can be expressed as a linear combination of these operators: 
\begin{eqnarray}
\nabla_a &=&g_a^{\;\;b}\; \nabla_b \nonumber \\
&=& (n_al^b+l_an^b-\bar{m}_am^b-m_a\bar{m}^b)\nabla_b  \nonumber \\
&=& n_aD + l_a \Delta - \bar{m}_a \delta -m_a \bar{\delta} \label{nablanp}
\end{eqnarray}
One can replace $\nabla_a$ by (\ref{nablanp}) and contract with the NP vectors to  convert all tensor equations ultimately to sets of scalar ones. Although this usually leads to a vast set of equations the fact that only scalars are involved and that most problems involve discrete symmetries makes them easier to handle. Evaluation of the derivatives requires evaluation of the derivative operators acting on basis spinors, hence we need to define the NP spin coefficients 
\begin{equation}
\begin{array}{ll}
Do_A=\epsilon o_A - \kappa \iota_A & D\iota_A=\pi o_A-\epsilon \iota_A \\
\Delta o_A=\gamma o_A -\tau \iota_A & \Delta \iota_A=\nu o_A -\gamma\iota_A \\
\delta o_A=\beta o_A-\sigma \iota_A & \delta \iota_A =\mu o_A -\beta \iota_A \\
\bar{\delta}o_A =\alpha o_A-\rho \iota_A & \bar{\delta} \iota_A=\lambda o_A-\alpha \iota_A
\end{array} \label{npscalars}
\end{equation}
From (\ref{npscalars}) one directly reads that $o^A D o_A=\kappa, o^A D i_A=\epsilon $ and so on, which is an alternative wave of defining NP spin coefficients. The action of the derivative operators on the null vectors follow from (\ref{nptetrad}) and (\ref{npscalars}).
\begin{equation}
\begin{array}{lll}
D l=(\epsilon + \bar{\epsilon})l  - \bar{\kappa}m- \kappa\bar{m} & D n=-(\epsilon + \bar{\epsilon})n  + \pi m + \bar{\pi}\bar{m} & D m=\bar{\pi}l- \kappa n +(\epsilon - \bar{\epsilon})m\\ 
\Delta l=(\gamma + \bar{\gamma})l -\bar{\tau}m - \tau\bar{m} & \Delta n=-(\gamma + \bar{\gamma})n +\nu m + \bar{\nu}\bar{m} & \Delta m=\bar{\nu}l- \tau n +(\gamma - \bar{\gamma})m \\
\delta l=(\bar{\alpha}+\beta)l-\bar{\rho}m -\sigma\bar{m} & \delta n=-(\bar{\alpha}+\beta)n +  \mu m + \bar{\lambda}\bar{m} & \delta m=\bar{\lambda}l -\sigma n +(\beta -\bar{\alpha})m\\ 
 & & \bar{\delta} m=\bar{\mu}l -\rho n +(\alpha -\bar{\beta})m
\end{array}\label{derivs0}
\end{equation}
Since we are only dealing with scalars for any NP quantitiy $\phi$, $\nabla_{[a}\nabla_{b]} \phi=0$. This leads to commutation relation for NP derivative operators.
\begin{eqnarray}
(\delta D-D\delta)\psi &=&[(\bar{\alpha}+\beta - \bar{\pi})D+\kappa \Delta - (\bar{\rho}+\epsilon-\bar{\epsilon} )\delta -\sigma \bar{\delta}]\psi \nonumber \\
(\Delta D- D\Delta)\psi &=& [(\gamma +\bar{\gamma} )D+ (\epsilon +\bar{\epsilon})\Delta - (\bar{\tau} +\pi)\delta -(\tau + \bar{\pi})\bar{\delta}]\psi \nonumber \\
(\delta \Delta -\Delta\delta)\psi &=& [-\bar{\nu} D+(\tau -\bar{\alpha} -\beta)\Delta +(\mu -\gamma + \bar{\gamma})\delta + \bar{\lambda}\bar{\delta}]\psi \nonumber \\
(\bar{\delta}\delta -\delta \bar{\delta})\psi &=& [(\bar{\mu}-\mu )D + (\bar{\rho} -\rho)\Delta +(\alpha - \bar{\beta})\delta - (\bar{\alpha} - \beta)\bar{\delta}]\psi  \label{npcommuters}
\end{eqnarray}
\section{Proca fields}\label{proca}
Proca equations describe massive vector fields. The Proca bivector $f_{\mu \nu}$ is defined in terms of the vector potential $A_{\nu}$ by $f_{\mu \nu}=\nabla_{\mu}A_{\nu}-\nabla_{\nu}A_{\mu}$. The field equations are simply
\begin{equation}
\nabla_{\mu} f^{\mu \nu}+ m^2A^{\nu}=0
\end{equation}
where $m$ is the mass term. There is no gauge freedom in Proca theory. The Lorentz condition $\nabla_{\mu}A^{\mu}=0$ is a consequence of the field equations. In Ricci flat space-times ($R_{\mu \nu}=0$) the Proca equation reduces to the massive wave equation
\begin{equation}
\nabla^b \nabla_b A^a + m^2 A^a=0
\label{procaeqn}
\end{equation}
In NP formalism $\nabla_b$ is given by (\ref{nablanp}), and by means of (\ref{spinvec}) the vector field can be expressed in the form:
\begin{equation}
A^a = \xi l^a + \eta n^a + \bar{\zeta} m^a + \zeta \bar{m}^a
\label{veca}
\end{equation}
The Proca equation (\ref{procaeqn}) consists of four equations which can be derived by multiplying it (from the left) with $l_a$, $n_a$, $m_a$, and $\bar{m}_a$, respectively. Let us check if one can obtain a decoupled equation.
\begin{proposition}
\label{prop1}
The Proca equation can be decoupled for one of the real scalars if and only if the background space-time contains a covariantly constant null vector. 
\end{proposition}
\begin{proof}
As we see in (\ref{derivs0}) the components of the directional derivatives of $l^a$ in the direction of $n^a$ vanish in general. If one requires that the directional derivatives of $m^a$ and $\bar{m}^a$ in the direction of $n^a$ also vanish, one can obtain a decoupled equation for $\eta$, since in that case,  terms proportional to $\zeta n^a$ or $\xi n^a$ will not occur in the explicit form of the first Proca equation. Now, consider a space-time that contains a covariantly constant null vector $\nabla l=0\Rightarrow Dl=\Delta l= \delta l=0$. This implies $\kappa=\sigma=\rho=\tau=0$. Using (\ref{derivs0}) one can also show that $(\epsilon + \bar{\epsilon})=(\gamma + \bar{\gamma})=(\bar{\alpha}+\beta)=0$, which will simplify the problem to some extent. The directional derivatives of $n$ and $m$  are now given by
\begin{equation}
\begin{array}{ll}
Dn=\pi m + \bar{\pi}\bar{m} & Dm=\bar{\pi}l + 2\epsilon m \\
\Delta n=\nu m + \bar{\nu}\bar{m} & \Delta m =\bar{\nu}l + 2\gamma m \\
\delta n= \mu m + \bar{\lambda}\bar{m} & \delta m =\bar{\lambda} - 2\bar{\alpha}m \\
\bar{\delta}n=\bar{\mu}\bar{m}+\lambda m & \bar{\delta}m=\bar{\mu}l + 2\alpha m
\end{array}
\label{derivs}
\end{equation}
Now we multiply the Proca equation (\ref{procaeqn}) from the left with $l_a$ to obtain a decoupled second order differential equation for $\eta$. 
\begin{eqnarray}
l_a(\nabla^b \nabla_b A^a &+& m^2 A^a)=[D\Delta+\Delta D- \delta \bar{\delta} -\ \bar{\delta} \delta \nonumber \\
&+& (\mu +\bar{\mu})D +(2\alpha - \pi)\delta + (2\bar{\alpha}- \bar{\pi})\bar{\delta}+ m^2]\eta=0
\label{etaeqn1}
\end{eqnarray}
(\ref{etaeqn1}) proves that decoupled equation  for $\eta$ can be found if $\kappa=\tau=\sigma=\rho=0$. Conversely, let us assume that one of the spin coefficients $\kappa$, $\tau$, $\sigma$, $\rho$ does not vanish. Then the derivatives of the terms $\bar{\zeta}m^a$ and $\zeta\bar{m}^a$ will produce terms proportional to $\bar{\zeta}n^a$ and $\zeta n^a$. When we multiply the Proca equation from the left with $l_a$, the resulting equation couples $\eta$ and $\zeta$. Thus decoupling is not possible.

Now let us choose a spin basis $(o,\iota)$ such that the covariantly constant null vector corresponds to $n^a$. In that case $\nabla n^a=0$, and the directional derivatives of the null vectors  $m^a$ and $\bar{m}^a$ have no components in the direction of $l^a$ since $\pi=\nu=\lambda=\mu=0$.  In that case, one can multiply the Proca equation from the left with $n_a$ to derive a decoupled equation for the other real scalar $\xi$.
\begin{eqnarray}
n_a(\nabla^b \nabla_b A^a &+& m^2 A^a)=[D\Delta+\Delta D- \delta \bar{\delta} -\ \bar{\delta} \delta \nonumber \\
&-& (\rho +\bar{\rho})\Delta +(2\alpha +\tau)\bar{\delta} + (2\bar{\alpha}+ \bar{\tau}){\delta}+ m^2]\xi=0
\label{xieqn1}
\end{eqnarray}
By the same argument (\ref{xieqn1}) will include $\zeta$ terms if at least one of the spin coefficients $\pi$, $\nu$, $\lambda$, $\mu$ does not vanish.
\end{proof}

In (\ref{veca}) we have considered the decomposition of the Proca vector over the NP tetrad. Let us also consider the decoupling problem for the Proca equation with respect to the components of the Proca vector in an orthonormal basis.

\begin{corollary}
Let $A^{\hat{i}}$ $(i=0,1,2,3)$ be the components of the Proca vector $A^a$ in the orthonormal basis $e_{\hat{i}}$, and $l,n,m,\bar{m}$ be the null vectors which constitute the corresponding null tetrad. Then a decoupled equation for $
(A^{\hat{0}}- A^{\hat{3}})$ can be derived if $\nabla l =0$, and a decoupled equation for $(A^{\hat{0}} + A^{\hat{3}})$  can be derived if $\nabla n =0$.
\end{corollary}
\begin{proof}
The Proca vector is given by
\begin{equation}
A^a=A^{\hat{0}}e_{\hat{0}}^{\;a} + A^{\hat{1}}e_{\hat{1}}^{\;a} + A^{\hat{2}}e_{\hat{2}}^{\;a} + A^{\hat{3}}e_{\hat{3}}^{\;a} \label{veca1}
\end{equation}
The orthonormal tetrad and the NP tetrad are related by (see e.g. \cite{agr})
\begin{equation}
\begin{array}{ll}
e_{\hat{0}}=(1/\sqrt{2})(l+n) & e_{\hat{1}}=(1/\sqrt{2})(m + \bar{m}) \\
e_{\hat{2}}=(1/\sqrt{2})i(m-\bar{m}) & e_{\hat{0}}=(1/\sqrt{2})(l-n)
\end{array}
\label{orthotetrad}
\end{equation}
Using (\ref{orthotetrad}), one can compare (\ref{veca}) and (\ref{veca1}) to see that $\xi=(A^{\hat{0}} + A^{\hat{3}})$, $\eta=(A^{\hat{0}} - A^{\hat{3}})$, $\bar{\zeta}=(A^{\hat{1}} + iA^{\hat{2}})$, and $\zeta=(A^{\hat{1}} - iA^{\hat{2}})$. Then the result follows from proposition (\ref{prop1}).
\end{proof}

The explicit forms of the equations $n_a(\nabla^b \nabla_b A^a+m^2 A^a)=0$, $m_a(\nabla^b \nabla_b A^a +m^2 A^a)=0$, and $\bar{m}_a(\nabla^b \nabla_b A^a ++m^2 A^a)=0$ are given in appendix. The equation $m_a(\nabla^b \nabla_b A^a +m^2 A^a)=0$ couples $\eta$ and $\zeta$, and its complex conjugate $\bar{m}_a(\nabla^b \nabla_b A^a+m^2 A^a)=0$ couples $\eta$ and $\bar{\zeta}$. The equation $n_a(\nabla^b \nabla_b A^a +m^2 A^a)=0$ couples $\xi$, $\eta$, and $\zeta$. In principle, one can obtain a solution for $\eta$ from (\ref{etaeqn1}) and substitute it into the equation  $m_a(\nabla^b \nabla_b A^a +m^2 A^a)=0$ to derive a decoupled equation for $\zeta$. Finally the solutions for $\eta$ and $\zeta$ can be substituted into the equation $n_a(\nabla^b \nabla_b A^a +m^2 A^a)=0$ to derive a decoupled equation for $\xi$. Thus, a complete solution for Proca equation can --in principle-- be found via a solution for (\ref{etaeqn1}) in a space-time that contains a covariantly constant vector field.

The space-times that  contain a covariantly constant vector field are known as pp-wave space-times. These are exact solutions to Einstein's field equations compatible with radiation associated with  a classical massless field, and they are everywhere   of Petrov type N~\cite{petrov} (see e.g. \cite{agr}). In all space-times of type N the repeated principal null direction generates a geodesic, shear free null congruence $(\kappa=\sigma=0)$, however all type N space-times do not contain a covariantly constant null vector. Decouplings (\ref{etaeqn1}) and (\ref{xieqn1}) are only possible in pp-wave space-times. They cannot be achieved in type N space-times other than pp-waves, that do not contain a covariantly constant null vector, or in algebraically special space-times of different types (including type D space-times such as Schwarzschild and Kerr), or in  
algebraically general space-times.

The pp-wave metrics have been widely studied both in the context of super-gravity and string theory since they constitute a convenient classical background and a simple toy model with vanishing curvature invariants, yet a rich internal structure. (see~\cite{coleman,aichelburg,guven78,dereli79,trautman,beler1,beler2,dereli86,amati,horowitz,berg,hubeny1,hubeny2,cariglia,dereli2011,gurses2014,baykal} and references therein) Plane wave metrics are subsets of pp-waves with extra planar symmetry along the wave fronts. The interest in plane waves has remarkably increased since Penrose discovered that one can associate a plane wave metric to every space-time and a choice of null geodesic in that space-time~\cite{penroselim}. This property is known as the Penrose limit, and was later extended to different string theories by G\"{u}ven~\cite{guven87,guven2000}. 

Ehlers and Kundt~\cite{ehlers} showed that if a vacuum space-time contains a covariantly constant vector field (i.e. if it is a pp-wave vacuum), coordinates $(u,v,x,y)$ can be found so that the line  element is given by (also see~\cite{agr}) 
\begin{equation}
ds^2=2H(u,x,y)du^2 + 2du dv - dx^2 - dy^2 
\label{ppmetric}
\end{equation}
The NP tetrad for (\ref{ppmetric}) is given by
\begin{equation}
l^a=(0,1,0,0), \quad n^a=(1,-H,0,0,), \quad m^a=-2^{-1/2}(0,0,1,i)
\label{pptetrad}
\end{equation}
The NP derivative operators have the form
\begin{equation}
\begin{array}{ll}
D=\partial /\partial v & \Delta=\partial/\partial u - H \partial/\partial v \\
\delta= (-1/\sqrt{2})(\partial /\partial x +i \partial /\partial y) & \bar{\delta}= (-1/\sqrt{2})(\partial /\partial x - i \partial /\partial y)
\end{array}
\end{equation}
The only non-vanishing NP spin coefficient for pp-wave vacuum space-time (\ref{ppmetric}) is $\nu$. The proca equations reduce to
\begin{eqnarray}
&&[D\Delta+\Delta D- \delta \bar{\delta} -\ \bar{\delta} \delta + m^2]\eta=0
\label{etasep1}\\
&&[D\Delta+\Delta D- \delta \bar{\delta} -\ \bar{\delta} \delta + m^2]\xi + \{ [2\nu D + D(\nu)]\zeta + \rm{c.c.} \}=0 \label{etasep2}\\
&&[ \delta \bar{\delta} +\ \bar{\delta} \delta -D\Delta - \Delta D-  m^2]\zeta - D(\eta \bar{\nu}) - \bar{\nu}D\eta = 0 \label{etasep3}\\
&&[ \delta \bar{\delta} +\bar{\delta} \delta  -D\Delta - \Delta D-  m^2]\bar{\zeta} - D(\bar{\eta}\nu) - \nu D\bar{\eta} = 0
\end{eqnarray}
Since the only non-vanishing spin coefficient is $\nu$, all the commutators except $(\delta, \Delta )$ vanish. 
\begin{eqnarray}
(D \Delta )\phi & =& (\Delta D) \phi = \frac{\partial^2 \phi}{\partial u \partial v}-H\frac{\partial^2 \phi}{\partial v^2}  \nonumber \\
(\delta \bar{\delta})\phi &=&(\bar{\delta}{\delta})\phi = \frac{1}{2} \left ( \frac{\partial^2 \phi}{\partial x^2} + \frac{\partial^2 \phi}{\partial y^2} \right )
\label{deltad}
\end{eqnarray}
Having $D\Delta=\Delta D$ and $\delta \bar{\delta} =\ \bar{\delta} \delta$, the decoupled equation for $\eta$ can be written as $[2(D\Delta - \delta \bar{\delta}) + m^2]\eta=0$. Letting $m=\sqrt{2}\mu$) it takes the form
\begin{equation}
\left( \frac{\partial^2 }{\partial u \partial v} -H \frac{\partial^2 }{\partial v^2} - \frac{1}{2} \frac{\partial^2 }{\partial x^2} - \frac{1}{2} \frac{\partial^2 }{\partial y^2} + \mu^2 \right) \eta=0
\label{phisepproca}
\end{equation}
The separability of the decoupled equation (\ref{phisepproca}) depends on the explicit form $H=H(u,x,y)$. Let us consider a simple textbook example of a sandwich wave with $H=(1/2)f(u) (x^2 - y^2)$~\cite{rindler} (also see \cite{agr} ) such that $f(u)$ is constant in the interval $(u_0,u_1)$, and vanishes outside.
Note that $\partial/\partial v$ is a Killing vector  impose separation of variables in the form $\eta=e^{ikv}U(u)X(x)Y(y)$. (\ref{phisepproca}) takes the form
\begin{equation}
ik \frac{(U^{'})}{U} + k^2 x^2 - \frac{1}{2}\frac{X^{''}}{X} - k^2 y^2 - \frac{1}{2}\frac{Y^{''}}{Y} +\mu^2 =0
\label{phisep1dir}
\end{equation}
where a prime denotes the derivatives of the functions with respect to relevant variables. The separated equations for the functions $U,X,Y$ are given by
\begin{eqnarray}
& & ik U^{'} - (\lambda_u - \mu^2) U=0 \\
&& X^{''} - 2(k^2x^2 + \lambda_x )X=0 \\
&& Y^{''} + 2(k^2y^2 + \lambda_y )Y=0
\end{eqnarray}
The functions have the form
\begin{equation}
\begin{array}{l}
U = e^{-iu(\lambda- \mu^2)/k} \\ 
X = C_{1} D_{v_1}(k_{1} x) + C_{2} D_{v_2}(i k_{1} x) \\
Y = C_{3} D_{v_3}((1+i) k_{2} y) + C_{4} D_{v_4}((-1+i) k_{2} y)
\end{array} \label{functions}
\end{equation}
where $k_{1} = 2^{3/4} \sqrt{k}$, $k_{2} = 2^{1/4} \sqrt{k}$, and $D_v(x)$ are parabolic cylinder functions. (see \cite{parabolic})
\section{Electromagnetic fields}
In this section we evaluate electromagnetic fields in pp-wave background. Let us start with the spinor equivalent of the Maxwell tensor in NP formalism.
\begin{equation}
F_{ABA'B'}=\phi_{AB} \epsilon_{A'B'}+\epsilon_{AB} \bar{\phi}_{A'B'}
\end{equation}
where $\phi_{AB}$ is a symmetric valence 2 spinor which generates 3 complex scalars via
\begin{equation}
\phi_0 =\phi_{AB}o^A o^B, \quad \phi_1 =\phi_{AB}o^A \iota^B, \quad \phi_2 =\phi_{AB}\iota^A \iota^B \label{npscalarsmax}
\end{equation}
The explicit forms of the source-free Maxwell equations  in terms of electromagnetic scalars are  as follows:
\begin{eqnarray}
& &(D-2\rho)\phi_1 -(\bar{\delta} +\pi-2\alpha)\phi_0 +\kappa \phi_2=0 \label{maxwell1}\\
& &(D-\rho +2\epsilon)\phi_2 -(\bar{\delta} +2\pi)\phi_1 +\lambda \phi_0 =0 \label{maxwell2} \\
& &(\delta -2\tau)\phi_1 -(\Delta +\mu -2\gamma)\phi_0 -\sigma \phi_2 =0 \label{maxwell3}\\
& &(\delta -\tau +2\beta)\phi_2 -(\Delta +2\mu)\phi_1 -\nu \phi_0 =0 \label{maxwell4}
\end{eqnarray}
First we require that the background space-time contains a covariantly constant null vector $\nabla l=0$ so that $(\kappa=\tau=\sigma=\rho=0)$. (\ref{maxwell1}) and (\ref{maxwell3}) reduce to
\begin{eqnarray}
& &D\phi_1 -\bar{\delta}\phi_0  =(\pi-2\alpha)\phi_0  \label{maxwell1n}\\
& &\Delta\phi_0 -\delta\phi_1  =(2\gamma-\mu)\phi_0 \label{maxwell3n}
\end{eqnarray}
Let us act on (\ref{maxwell1n}) with $\delta$, and on (\ref{maxwell3n}) with  $D$ from the left. 
\begin{equation}
(\delta D -D \delta)\phi_1 +D \Delta \phi_0 -\delta \bar{\delta}\phi_0=\delta [(\pi -2\alpha)\phi_0]+D[(2\gamma-\mu)\phi_0]
\label{phi}
\end{equation}
In that background the commutation relation for $\delta$ and $D$ acting on a scalar is also reduced to
\begin{equation}
(\delta D -D \delta)\phi_1=[ -\bar{\pi}D  -(\epsilon - \bar{\epsilon})\delta ]\phi_1
\label{commut}
\end{equation}
From Maxwell equations (\ref{maxwell1n}) and (\ref{maxwell3n}) we have $D\phi_1= [\bar{\delta}+(\pi-2\alpha)]\phi_0$ and $\delta\phi_1  =[\Delta-(2\gamma-\mu)]\phi_0$. Thus, using (\ref{commut}) with (\ref{maxwell1n}) and (\ref{maxwell3n}) we can transform (\ref{phi}) to a decoupled second order equation for $\phi_0$
\begin{eqnarray}
& & \{ D\Delta -\delta \bar{\delta} - \bar{\pi}[\bar{\delta}+(\pi-2\alpha)] -(\epsilon - \bar{\epsilon})[\Delta-(2\gamma-\mu)] \}\phi_0 \nonumber \\
& & -\delta[(\pi-2\alpha)\phi_0] -D[(2\gamma-\mu )\phi_0] =0
\label{phidec}
\end{eqnarray}
In pp-wave vacuum space-time (\ref{ppmetric}) where the only non-vanishing NP scalar is $\nu$,\footnote{One can apply a null rotation around $l$ and the spin coefficients $\mu$, $\lambda$, and $\pi$ can attain non-zero values. In that case (\ref{phidec}) remains valid. Such a transformation is possible, though it may not be useful.}  the decoupled equation (\ref{phidec}) is reduced to
\begin{equation}
(D\Delta -\delta \bar{\delta})\phi_0=0
\label{phidec1}
\end{equation}
Relatively simple, but not decoupled second order equations can also be derived for $\phi_1$ and $\phi_2$.
\begin{eqnarray}
&&(D\Delta -\delta \bar{\delta})\phi_1 + D(\nu \phi_0)=0 \\
&& (\Delta D- \bar{\delta} \delta)\phi_2 -\nu D\phi_1 +\bar{\delta}(\nu \phi_0)=0
\end{eqnarray}
Using (\ref{deltad}),  (\ref{phidec1}) takes the form
\begin{equation}
\left( \frac{\partial^2 }{\partial u \partial v} -H \frac{\partial^2 }{\partial v^2} - \frac{1}{2} \frac{\partial^2 }{\partial x^2} - \frac{1}{2} \frac{\partial^2 }{\partial y^2}\right) \phi_0=0
\label{phisep}
\end{equation}
(\ref{phisep}) is the decoupled differential equation or the Maxwell scalar $\phi_0$ in pp-wave background. We see that it is identical with the decoupled equation for $\eta$ in the Proca case as we let $\mu=0$. Let us consider the  sandwich wave example of the previous section and impose separation of variables in the form $\phi_0=e^{ikv}U(u)X(x)Y(y)$. (\ref{phisep}) takes the form
\begin{equation}
ik \frac{(U^{'})}{U} + k^2 x^2 - \frac{1}{2}\frac{X^{''}}{X} - k^2 y^2 - \frac{1}{2}\frac{Y^{''}}{Y} =0
\label{phisep1}
\end{equation}
The separated equations for the functions $U,X,Y$ are given by
\begin{eqnarray}
& & ik U^{'} - \lambda_u U=0 \\
&& X^{''} - 2(k^2x^2 + \lambda_x )X=0 \\
&& Y^{''} + 2(k^2y^2 + \lambda_y )Y=0
\end{eqnarray}
The functions have the form (\ref{functions}) with $\mu=0$.
\section{Dirac fields}
We start with Dirac equation which couples two fermion fields via
\begin{eqnarray}
& &\nabla_{A A'}P^{A} + i \mu \bar{Q}_{A'}=0 \nonumber \\
& & \nabla_{A A'}Q^{A} +i \mu \bar{P}_{A'}=0   \label{diraceqn} 
\end{eqnarray}
In NP formalism, Dirac's equations (\ref{diraceqn}) can explicitly be written in the form:
\begin{eqnarray}
& &(D+\epsilon-\rho)P^0+ (\bar{\delta}+\pi -\alpha)P^1=i\mu \bar{Q}^{\dot{1}} \label{dirac1} \\
& &(\Delta+\mu -\gamma )P^1 + (\delta -\tau +\beta )P^0=-i\mu \bar{Q}^{\dot{0}} \label{dirac2} \\
& & (D+ \bar{\epsilon} -\bar{\rho} )\bar{Q}^{\dot{0}}+ (\delta +\bar{\pi} -\bar{\alpha})\bar{Q}^{\dot{1}}=-i\mu P^1 \label{dirac3} \\
& & (\Delta +\bar{\mu} - \bar{\gamma})\bar{Q}^{\dot{1}} +(\bar{\delta} + \bar{\beta} - \bar{\tau})\bar{Q}^{\dot{0}}=i \mu P^0 \label{dirac4}
\end{eqnarray}
where $P^0$,$Q^0$ and $P^1$,$Q^1$ are components of $P^A$,$Q^A$ along the spinor dyad basis $o^A$ and $\iota^A$ respectively. 
In pp-wave background the explicit forms of Dirac equations reduce to
\begin{eqnarray}
& &DP^0+ \bar{\delta}P^1=i\mu\bar{Q}^{\dot{1}} \label{dirac1pp} \\
& &\Delta P^1 + \delta P^0=- i\mu \bar{Q}^{\dot{0}} \label{dirac2pp} \\
& & D\bar{Q}^{\dot{0}}+ \delta \bar{Q}^{\dot{1}}=-i\mu P^1 \label{dirac3pp} \\
& & \Delta \bar{Q}^{\dot{1}} +\bar{\delta} \bar{Q}^{\dot{0}}=i \mu P^0 \label{dirac4pp}
\end{eqnarray}
Since the only non-vanishing commutator is $(\delta \Delta -\Delta\delta)= -\bar{\nu} D$ we can eliminate $P^0$ from (\ref{dirac1pp}) and (\ref{dirac2pp})
\begin{eqnarray}
&&(D\Delta - \delta \bar{\delta})P^1=-i \mu (D \bar{Q}^{\dot{0}}  + \delta \bar{Q}^{\dot{1}}) =-\mu^2 P^1 \nonumber \\ 
& \Rightarrow &(D\Delta - \delta \bar{\delta} + \mu^2)P^1=0 \label{p1dec}
\end{eqnarray}
Similarly one can derive a decoupled equation for $\bar{Q}^{\dot{1}}$.
\begin{equation}
(D\Delta - \bar{\delta} \delta  + \mu^2)\bar{Q}^{\dot{1}}=0 \label{q1dec}
\end{equation}
The equations for $P^0$ and $\bar{Q}^{\dot{0}}$ can not be decoupled.
\begin{eqnarray}
&& (\Delta D  - \bar{\delta}\delta  + \mu^2)P^0 + \nu D P^1=0  \\
&&(\Delta D - \delta \bar{\delta}   + \mu^2)\bar{Q}^{\dot{0}} + \bar{\nu} D \bar{Q}^{\dot{1}} =0
\end{eqnarray}
The decoupled equations for $P^1$ and $\bar{Q}^{\dot{1}}$ have exactly the same form  as (\ref{phisepproca}), while they reduce to (\ref{phisep}) for massless Dirac fields. The arguments for separation of variables for Proca and Maxwell fields  directly apply to this case.
\section{Summary and conclusions}
In this work we represented the Proca vector by one complex and two real scalars. We showed that a decoupled equation can be derived for one of the real scalars if and only if the background space-time contains a covariantly constant null vector; thus it must be a pp-wave. We derived the explicit form of the decoupled equation in vacuum pp-wave metric and applied separation of variables for a simple example of a sandwich wave. We proceeded with Maxwell and Dirac fields in pp-wave background and derived decoupled equations for Maxwell scalar $\phi_0$ and components in dyad legs $P^1$ and $\bar{Q}^{\dot{1}}$. We showed that the decoupled equations in pp-wave background have the same form and the arguments for separation of variables can be directly applied.

Note that there could be different decompositions of the Proca vector into degrees of freedom, some of which could satisfy decoupled equations in a wider class of space-times, hence our analysis in the NP formalism does \emph{not} exclude all decouplings in space-times other than pp-waves. On the other hand, the NP-formalism is well-motivated by the correspondence between the structure of matter/interaction fields and the structure of spacetime, i.e. spinors forming the most general representations of the Lorentz group and null vectors forming light cones;  hence the associated decomposition should be regarded as very natural in some sense, therefore the result about decoupling and separability should also be regarded relevant.
%\subsection{Subsection title}
%\label{sec:2}
%as required. Don't forget to give each section
%and subsection a unique label (see Sect.~\ref{sec:1}).
%\paragraph{Paragraph headings} Use paragraph headings as needed.
%\begin{equation}
%a^2+b^2=c^2
%\end{equation}

% For one-column wide figures use
%\begin{figure}
% Use the relevant command to insert your figure file.
% For example, with the graphicx package use
 % \includegraphics{example.eps}
% figure caption is below the figure
%\caption{Please write your figure caption here}
%\label{fig:1}       % Give a unique label
%\end{figure}
%
% For two-column wide figures use
%\begin{figure*}
% Use the relevant command to insert your figure file.
% For example, with the graphicx package use
 % \includegraphics[width=0.75\textwidth]{example.eps}
% figure caption is below the figure
%\caption{Please write your figure caption here}
%\label{fig:2}       % Give a unique label
%\end{figure*}
%
% For tables use
%\begin{table}
% table caption is above the table
%\caption{Please write your table caption here}
%\label{tab:1}       % Give a unique label
% For LaTeX tables use
%\begin{tabular}{lll}
%\hline\noalign{\smallskip}
%first & second & third  \\
%\noalign{\smallskip}\hline\noalign{\smallskip}
%number & number & number \\
%number & number & number \\
%\noalign{\smallskip}\hline
%\end{tabular}
%\end{table}

\begin{acknowledgements}
This work is supported by Bo\u{g}azi\c{c}i University Research Fund, by grant number 7981.
\end{acknowledgements}

\section*{Appendix}
The Proca equation is
\begin{equation}
[( n^bD + l^b \Delta - \bar{m}^b \delta -m^b \bar{\delta})( n_bD + l_b \Delta - \bar{m}_b \delta -m_b \bar{\delta})+m^2]A^a=0
\end{equation}
If the background space-time admits a covariantly constant null vector,  it takes the form ( using (\ref{ortog}) and (\ref{derivs}))
\begin{equation}
\{[D\Delta+\Delta D- \delta \bar{\delta} -\ \bar{\delta} \delta \nonumber 
+ (\mu +\bar{\mu})D +(2\alpha - \pi)\delta + (2\bar{\alpha}- \bar{\pi})\bar{\delta}]+m^2\}A^a=0
\end{equation}
Multiplying  from the left with $m_a$ we get 
\begin{eqnarray}
&&m_a \nabla^b \nabla_b A^a= (-D\Delta -\Delta D + \delta\bar{\delta} + \bar{\delta}\delta)\zeta \nonumber \\
&& + D(2\gamma \zeta)+ \Delta (2 \epsilon \zeta ) - \delta (2 \alpha \zeta ) + \bar{\delta}(2 \bar{\alpha} \zeta ) \nonumber \\
&& \{ [2\gamma - (\mu + \bar{\mu})]D + 2\epsilon \Delta -[2\alpha + (2\alpha - \pi)]\delta + \bar{\pi}\bar{\delta} \nonumber \\
&& -8\gamma \epsilon -8 \alpha \bar{\alpha} + 2\epsilon (\mu + \bar{\mu})+2\bar{\alpha}\pi -2\alpha \bar{\pi} \} \zeta \nonumber  \\
&& -D(\bar{\nu}\eta) - \Delta(\bar{\pi}\eta) + \delta(\bar{\mu}\eta) + \bar{\delta}(\bar{\lambda}\eta +\{ -\bar{\pi}\Delta - \bar{\nu}D + \bar{\lambda}\bar{\delta} + \bar{\mu}\delta  \nonumber \\
&& +2\epsilon \bar{\nu} + 2 \gamma \bar{\pi} - 4 \alpha \bar{\lambda} + 2 \bar{\pi} \bar{\mu} +\pi \bar{\lambda}-\bar{\pi} \mu \} \eta -m^2 \zeta=0
\label{zetaeta}
\end{eqnarray}
The equation $\bar{m}_a(\nabla^b \nabla_b+m^2) A^a=0$ is the complex conjugate of (\ref{zetaeta}). Let us evaluate $n_a(\nabla^b \nabla_b+m^2) A^a=0$
\begin{eqnarray}
&& n_a\nabla^b \nabla_b A^a=\{(D\Delta +\Delta D -\delta\bar{\delta} - \bar{\delta}\delta)+[\mu D + (2\alpha - \pi)\delta + \rm{c.c.}] \}\xi \nonumber \\
&&+\{[D(\nu)+2\nu D+ \Delta (\pi)+2\pi\Delta -\delta(\zeta)-2\lambda\delta -\bar{\delta}(\mu)-2\mu\bar{\delta}]\zeta +\rm{c.c.}\}\nonumber \\
&& +[(2\epsilon \bar{\nu}-2\gamma \pi + 4\alpha \mu +\bar{\mu}\pi - \bar{\pi}\lambda)\zeta + \rm{c.c.}] \nonumber \\
&&+2(\bar{\pi}\nu +\bar{\nu}\pi -\mu\bar{\mu} -\lambda \bar{\lambda})\eta +m^2 \xi=0
\end{eqnarray}
These equations will take the form (\ref{etasep2}) and (\ref{etasep3}) in a pp-wave space-time.
In a generic space-time the equation $n_a(\nabla^b \nabla_b+m^2) A^a=0$ takes the form:
\begin{eqnarray}
&&\{ (D\Delta +\Delta D -\delta\bar{\delta} - \bar{\delta}\delta) + [(\gamma + \bar{\gamma}) + (\mu + \bar{\mu}) ]D  + [3(\epsilon + \bar{\epsilon}) - ( \rho + \bar{\rho} )] \Delta \} \xi \nonumber \\
&& \{[ \bar{\tau} - \pi - (\alpha + 3\bar{\beta})]\delta + \rm{c.c} \} \xi \nonumber \\
&&\{ D (\gamma + \bar{\gamma}) + \Delta (\epsilon + \bar{\epsilon}) - [\delta (\alpha + \bar{\beta}) +\rm{c.c} ] \} \xi \nonumber \\
&& \{  (\epsilon + \bar{\epsilon})[2 (\gamma + \bar{\gamma}) + (\mu + \bar{\mu}) ] + [(\alpha + \bar{\beta})(\tau - \bar{\pi} -2\beta) + \rm{c.c} ]  \nonumber \\
&& + (\rho \mu + \sigma \lambda - \tau \pi - \kappa \nu + \rm{c.c} )-(\rho + \bar{\rho})(\gamma + \bar{\gamma})  + m^2 \} \xi \nonumber \\
&& \{ [\nu(3\bar{\epsilon} + \epsilon - \rho -\bar{\rho})  \mu(\alpha - 3 \bar{\beta} + \bar{\tau}) + \lambda (\tau - \bar{\pi} - \bar{\alpha} - \beta)+ \pi(2\bar{\gamma} + \bar{\mu}) \nonumber \\
&& +D(\nu ) + \Delta (\pi ) - \delta (\lambda) - \bar{\delta} \mu ] \zeta + \rm{c.c} \} \nonumber \\
&& +2[(\nu \bar{\pi} +  \bar{\nu}\pi)  - (\mu \bar{\mu} + \lambda \bar{\lambda})]\eta =0 \label{xigen}
\end{eqnarray}
The equation (\ref{xigen}) couples $\xi$, $\eta$, $\zeta$ and $\bar{\zeta}$. If the space-time satisfies $\lambda=\nu=\pi=\mu=0$, it reduces to a decoupled equation for $\xi$ as suggested by proposition (\ref{prop1}). If we further impose $(\epsilon + \bar{\epsilon})=(\gamma + \bar{\gamma})=(\bar{\alpha}+\beta)=0$ (so that $\alpha + 3\bar{\beta}= -2\alpha$ ), (\ref{xigen}) reduces to (\ref{xieqn1}).

Let us also review the form of  $l_a(\nabla^b \nabla_b+m^2) A^a=0$ in a generic space-time.
\begin{eqnarray}
&&\{ (D\Delta +\Delta D -\delta\bar{\delta} - \bar{\delta}\delta)+[-3(\gamma + \bar{\gamma}) + (\mu + \bar{\mu}) ]D + [-(\epsilon + \bar{\epsilon}) + ( \rho + \bar{\rho} )] \Delta \} \eta \nonumber \\
&& \{[ \bar{\tau} - \pi + (3\alpha + \bar{\beta})]\delta + \rm{c.c} \} \eta \nonumber \\
&&\{ -D (\gamma + \bar{\gamma}) - \Delta (\epsilon + \bar{\epsilon}) +[\delta (\alpha + \bar{\beta}) +\rm{c.c} ] \} \eta \nonumber \\
&& \{  (\epsilon + \bar{\epsilon})[2 (\gamma + \bar{\gamma}) - (\mu + \bar{\mu}) ] + [(\alpha + \bar{\beta})(\bar{\pi}- \tau -  2\bar{\alpha}) + \rm{c.c} ]  \nonumber \\
&& + (\rho \mu + \sigma \lambda - \tau \pi - \kappa \nu + \rm{c.c} ) + (\rho + \bar{\rho} )(\gamma + \bar{\gamma}) + m^2 \} \eta \nonumber \\
&&\{ [ \bar{\kappa}(3\gamma + \bar{\gamma} - \mu - \bar{\mu}) + \bar{\rho}(\bar{\beta}- 3\alpha + \pi ) + \bar\sigma(\bar{\pi} - \tau - \alpha - \bar{\beta}) + \bar{\tau}(\rho - 2 \bar{\epsilon}) \nonumber \\
&& -D (\bar{\tau}) - \Delta (\bar{\kappa}) +\delta (\bar{\sigma}) + \bar{\delta}(\bar{\rho} ) ] \zeta + \rm{c.c} \} \nonumber \\
&&+ 2[(\kappa \bar{\tau} + \bar{\kappa}\tau) - (\rho\bar{\rho} + \sigma \bar{\sigma})]\xi=0 \label{etagen}
\end{eqnarray}
Similarly, the equation (\ref{etagen}) couples $\xi$, $\eta$, $\zeta$ and $\bar{\zeta}$. If the space-time satisfies $\kappa=\tau=\rho=\sigma=0$, it reduces to a decoupled equation for $\eta$ as suggested by proposition (\ref{prop1}). If we further impose $(\epsilon + \bar{\epsilon})=(\gamma + \bar{\gamma})=(\bar{\alpha}+\beta)=0$ (so that $3\alpha + \bar{\beta}=2\alpha$ ), (\ref{etagen}) reduces to (\ref{etaeqn1}).

% BibTeX users please use one of
%\bibliographystyle{spbasic}      % basic style, author-year citations
%\bibliographystyle{spmpsci}      % mathematics and physical sciences
%\bibliographystyle{spphys}       % APS-like style for physics
%\bibliography{}   % name your BibTeX data base

% Non-BibTeX users please use

\end{document}